\documentclass[a4paper,twoside]{amsart}
\usepackage{amsmath}
\usepackage{amsfonts}
\usepackage{amssymb}
\usepackage{amsthm}
\usepackage{newlfont}
\usepackage{graphicx}
\usepackage{amscd}
\usepackage{braket}
\usepackage{arydshln}

\usepackage{xcolor}
\usepackage[normalem]{ulem}

\theoremstyle{plain}
\newtheorem{Th}{Theorem}[section]
\newtheorem{Cor}[Th]{Corollary}

\newtheorem{Prop}[Th]{Proposition}
\theoremstyle{definition}

\newtheorem{Ex}{Example}[section]
\theoremstyle{remark}
\newtheorem*{Rem}{Remark}%[section]
\numberwithin{equation}{section}
\newcommand{\BB}{{\mathbb B}}
\newcommand{\HH}{{\mathbb H}}

\newcommand{\NN}{{\mathbb N}}

\newcommand{\CC}{{\mathbb C}}

\newcommand{\II}{{\mathbb I}}
\newcommand{\ZZ}{{\mathbb Z}}

\begin{document}

\title[On aggregation--quantization permutability problem]
{On aggregation--quantization permutability problem\\ for discrete-time Markov chains}

\author{Adam Doliwa$^*$}

\author{Artur Siemaszko}

\author{Adam Zalewski}

\address{Faculty of Mathematics and Computer Science\\
	University of Warmia and Mazury in Olsztyn\\
	ul.~S{\l}oneczna~54\\ 10-710~Olsztyn\\ Poland} 
\email{adam.doliwa@uwm.edu.pl, artur@uwm.edu.pl, adam.zalewski@uwm.edu.pl}
%\urladdr{http://wmii.uwm.edu.pl/~doliwa/}

%
\date{}
\let\thefootnote\relax\footnote{$^*$Corresponding author}

\begin{abstract}
Given random walk on a graph, the corresponding discrete-time quantum walk can be constructed using the method proposed by Szegedy. On the other hand, given a partition of the set of states of a Markov chain, one can study the corresponding aggregated process. We extend the aggregation technique to the level of quantum Markov chains. We provide conditions under which application of these two operations -- Szegedy's quantization and aggregation -- give the same result.  These consist of the strong lumpability of the underlying Markov chain supplemented by nonlinear relations, generalizing Kolmogorov's cycle condition, that account for the coefficients of both the original and the aggregated transition matrices. In particular, we show that the conditions are satisfied in the case of the random walk on graphs equipped with equitable partitions. We present several examples, which include the classical/quantum walks on Platonic solids. We discuss also relation of discrete-time classical/quantum walks on $N$-dimensional hypercube and the Ehrenfests urn model with $N$ particles. We apply our technique for of discrete-time walks on Cayley graphs of free groups. We also compare our results with those obtained using Cantero--Moral--Vel\'{a}zquez  uniformization of unitary matrices.

\end{abstract}
\maketitle

\section{Introduction}

Random walks form one of the cornerstones in modern theoretical computer science~\cite{MitzenmacherUpfal}. 
However, even seemingly simple models can be challenging due to the large number of states, making analysis a difficult task. To tackle this problem, one can reduce the state space of the underlying Markov chain by aggregating states that exhibit equivalent behaviors, see for example \cite{GangulyPetrovKoeppl} for application of the technique in mathematical biology, or \cite{Wilson} for similar ideas in theoretical physics. 
In the theory of Markov chains this state-based reduction technique is known as lumping. Various notions of lumping, including strong and weak lumping~\cite{KemenySnell}, exact lumping~\cite{Schweitzer}, and strict lumping~\cite{Buchholz}, have been introduced in the literature.  Information-theoretic approach to the reduction of Markov chains with special emphasis on their application in data mining and knowledge discovery is the subject of~\cite{Geiger}.

In quantum computing~\cite{Hirvensalo,NielsenChuang} the  measurement-induced randomness is inherent and, together with entanglement and interference, provides the three pillars of quantum algorithms. The famous Grover's search algorithm~\cite{Grover}, based on the amplitude amplification, can be put within the quantum walk on a finite graph, i.e. the quantum analogue of a Markov chain~\cite{ShenviKempeWhaley,Szegedy,Portugal}. Recent works show the fundamental role of quantum walks in the domain of quantum algorithms~\cite{1sQEW,Childs,Kempe,Kendon,MNRS,QW-review}. 
Any quantum algorithm can be expressed as a quantum walk, and thus described as a moving quantum particle on a graph. This opens the possibility of using quantum walks as a generic framework to design quantum algorithms.

Quantum walks come in two variants: discrete- and continuous-time, and both have been shown to constitute a universal model of quantum computation~\cite{Childs,LCETK}. In this work we focus on the discrete-time model with coins, which consists of a bipartite quantum state, containing information about the position and coin states of a walker, and a unitary evolution operator, composed of the subsequent application of a coin operator, and a swap operator. Discrete-time quantum walks on graphs using a coin operator were initiated in~\cite{Aharonov-Vazirani}.  A general quantization procedure for discrete-time finite Markov chains has been proposed by Szegedy~\cite{Szegedy}. His approach uses Grover's-type coin operator but discrete Fourier transform-type coin, especially in the case of regular graphs, is used often as well~\cite{Portugal}. Complex-phase extensions of Szegedy quantum walk on graphs are studied in \cite{Ortega-Delgado}. Application of Szegedy's construction to random walks on half-line showed~\cite{Doliwa-Siemaszko-QW} its connection with orthogonal polynomials on the unit circle \cite{Simon-OPUC} generalizing the approach initiated in~\cite{CMGV}.
Quantum walks on certain Cayley graphs, including graphs generated by free groups were investigated in~\cite{AcevedoGobron}. In \cite{MoorRussell} the Grover walk on the hypercube was analyzed using the Fourier method. Discrete quantum walks on the symmetric group was studied in~\cite{Banerjee}.  We also mention, that open quantum walks taking into account the interaction with environment are also being investigated~\cite{APSS}.

Given applicability of coarse-graining techniques for Markov chains, in the present work we study the corresponding operation on the level of Szegedy's quantizations of the underlying processes. In particular we define quantum version of the aggregation procedure.  Given a graph $G(V,E)$ and given a partition $\tilde{V}$ of its vertex set $V$ we build in a natural way the coarsened graph $\tilde{G}(\tilde{V},\tilde{E})$, see \eqref{eq:tilde-E} for details.
Given random walks on $G(V,E)$ and $\tilde{G}(\tilde{V},\tilde{E})$ with transition matrices $P$ and $\tilde{P}$, respectively, Szegedy's construction provides the corresponding evolution operators $U$ and $\tilde{U}$. On the other hand, the partition $\tilde{V}$ allows to restrict Hilbert space of the quantum walk on $G$ to the subspace spanned by the aggregated states, as defined by formula \eqref{eq:aggregated-uv}. We impose the following aggregation--quantization condition: \emph{the evolution operator $U$ of the quantized walk on $G(V,E)$ should act on the aggregated states in the same way as the operator $\tilde{U}$  act on the corresponding states of the quantum walk on the aggregated graph $\tilde{G}(\tilde{V},\tilde{E})$.
 } 

The main subject of our work is to study conditions which allow for the aggregation--quantization permutability.  We find that the strong lumpability of the random walk on $G(V,E)$ with respect to the partition $\tilde{V}$ is necessary for the permutability. In addition, the transition probabilities of the original walk and its aggregation should satisfy certain nonlinear relations. They originate as compatibility conditions of a system of linear equations which allow to calculate coefficients in definition of the aggregated states. Interestingly, the linear equations are of the form of the detailed balance condition, known from the theory of reversible Markov chains~\cite{Kelly}, and the nonlinear equations generalize the corresponding Kolmogorov's cycle condition.    

We provide also several examples how our theory works. In doing that we concentrate mostly on graphs with symmetries, where quantum evolution naturally confines the dynamics to a smaller subspace of the full Hilbert space~\cite{Ekert2} and  provides aggregation of the underlying process. 
In particular, we show that our conditions are satisfied in the case of the random walk on graphs equipped with equitable partitions. Our examples start with the classical/quantum walks on Platonic solids. We discuss also relation of discrete-time classical/quantum walks on $N$-dimensional hypercube and the Ehrenfests urn model with $N$ particles. We apply our technique for of discrete-time walks on Cayley graphs of free groups. 

We do not know any previous work where the aggregation--quantization permutability problem for Markov chains was formulated and systematically investigated. However,  application of the coarse-graining techniques to study various problems related to quantum walks, both continuous- or discrete-time, is not new. We would like to point out \cite{KroviBrun}
where discrete-time quantum walks on quotient graphs resulting from symmetric graphs with suitably chosen initial state were investigated. The quotient graph is obtained from the original graph by identifying vertices and edges which form an orbit under the action of a subgroup of the automorphism group of the graph; we remark that in our approach we do not assume symmetry of the graph what allows for non-uniform aggregation coefficients. In~\cite{Godsil-transfer} more general aggregation procedure, based on equitable partitions, in the context of perfect state transfer via continuous-time quantum walk was proposed. Quantum walk based search algorithms in relation to partition of graphs, especially equitable partitions, were considered in~\cite{Ide}.

We compare our approach with that coming from the uniformization of unitary matrices to the Cantero--Moral--Vel\'{a}zquez (CMV) form~\cite{CMV-1,KillipNenciu-CMV,Simon-CMV}, when a suitable basis adjusted to both the quantum evolution and the initial state respecting the symmetry can be constructed. It is known that every cyclic unitary operator is unitarily equivalent to a unique CMV matrix. This makes CMV matrices the proper canonical form for unitary operators, just as Jacobi matrices are for self-adjoint operators. This analogy has become particularly important in the spectral analysis of Szegedy and coined quantum walks, where CMV methods often play the same role that Jacobi matrices play for continuous-time walks~\cite{CMGV,Doliwa-Siemaszko-QW}. The CMV basis construction applied to a discrete-time quantum walk and an initial state vector, naturally selects an appropriate subspace of the Hilbert space in which the evolution takes place, and has the form of a walk on the segment. For this reason the CMV uniformization and coarse-graining aggregation, although conceptually different, are two ways of simplifying the analysis of a quantum walk on a graph.

The structure of our article is as follows. In Section~\ref{sec:preliminaries}
we first recall basic information on aggregation of Markov chains, their Szegedy's quantization, and the Cantero--Moral--Vel\'{a}zquez  theory. We illustrate these basic concepts on the example of the hexahedron graph. Then in Section~\ref{sec:qa} we formulate the aggregation operation for quantum Markov chains, and we study its permutability with Szegedy's quantization procedure. In particular, we present  conditions on transition probability matrix of a given Markov chain in relation to its aggregation, which makes the two operations commute. As basic examples of quantum aggregation we present walks on graphs of other Platonic solids. In next two Sections we present more advanced examples of quantum aggregation of homogeneous random walk an the $N$-dimensional hypercube in relation to the Ehrenfests urn model, and the reduction of Szegedy's quantization of the random walk on free groups. The last Section concludes our presentation and points out towards various generalizations and new research directions.

\section{Preliminaries}
\label{sec:preliminaries}
Consider a directed graph $G = (V,E)$ with vertex set $V$ and edge set $E\subset V\times V$. Given a partition $\tilde{V}$ of its vertex set (into disjoint nonempty subsets, whose union is $V$) define the coarsened graph $\tilde{G} = (\tilde{V}, \tilde{E})$ with the edge set $ \tilde{E} \subset \tilde{V}\times \tilde{V}$ defined by 
\begin{equation} \label{eq:tilde-E}
	(u,v)\in \tilde{E} \quad \Leftrightarrow \quad \exists \; i\in u, \, j\in v \quad \text{such that} \quad  (i,j)\in E.
\end{equation}
In machine learning, graph coarsening goes under various names, e.g., graph
down-sampling or graph reduction. Its goal in most cases is to replace some original graph by one which has fewer nodes, but whose structure and characteristics are similar to those of the original graph.

\subsection{Lumpability of Markov chains}
Let $X = (X_n)_{n\geq 0}$ be a discrete-time Markov chain~\cite{Norris} taking values in a finite set  $V$, characterized by its transition matrix $P = (P_{ij})_{i,j\in V}$ and initial distribution $\pi$. Entry $P_{ij}$ represents the probability of making a transition from $i$ to $j$, in particular $\sum_{j=1}^N P_{ij} = 1$. Such a matrix gives rise to a directed weighted graph $G(V,E)$, where the ordered pair $(i,j)$ of vertices belongs to the edge set $E$ if and only if $P_{ij}>0$. 

Let $\tilde{V}$ be a fixed partition of the set $V$ of micro-states. We use it to construct an aggregated process $\tilde{X}$ by substituting each subset $u\subset V$ of states by a single macro-state
\begin{equation}
	\tilde{X}_n = u  \Leftrightarrow X_n \in u, \qquad \text{for any} \quad n .
\end{equation}
Correspondingly, for any probability distribution $\pi$ on $V$, define the distribution $\tilde{\pi} = (\tilde{\pi}_u)_{u\in \tilde{V}}$ on $\tilde{V}$ by
\begin{equation}
\tilde{\pi}_u = \sum_{i\in u} \pi_i .
\end{equation}

When the aggregated process is Markov  for any initial distribution then $X$ is said to be \emph{lumpable}~\cite{KemenySnell} (with respect to a given partition). 
It is known  that the lumpability can be characterized
on the level of the transition matrix $P$ such that sums of transition probabilities from each state in a partition group to all states of another partition have to be equal (row sum criterion), and then the quantity is equal the to the corresponding element of the transition matrix of the aggregated chain
\begin{equation} \label{eq:lumpability}
	\sum_{j\in v} P_{ij} = \tilde{P}_{uv}, \qquad i\in u, \qquad u,v \in \tilde{V}.
\end{equation}
It was noted in \cite{BarrThomas} that the spectrum of $\tilde{P}$ is then a subset of the spectrum of $P$. 
\begin{Ex} \label{ex:hexahedron-E}
	Consider discrete-time random walk on the hexahedron graph $H$, visualized on Figure~\ref{fig:hexahedron},
	\begin{figure}[h!]
		\begin{center}
			\includegraphics[width=5cm]{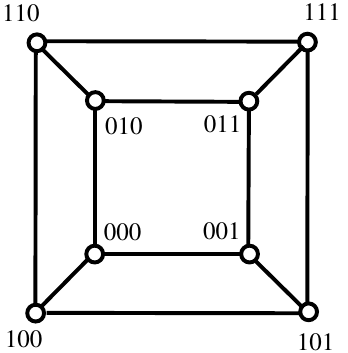}
		\end{center}
		\caption{The hexahedron graph}
		\label{fig:hexahedron}
	\end{figure}
	with vertices indexed by binary sequences of length $3$, an edge connects vertices whose sequences differ exactly at one place. 
	We group vertices with the same distance from the vertex $000$, what allows to consider aggregated process on the set $\{ A, B, C, D\}$, where 
\begin{equation}
	A = \{ 000 \}, \quad 
	B = \{ 001, 010, 100 \}, \quad C = \{ 011, 101, 110 \}, \quad D = \{ 111 \}. 
\end{equation} 	
The transition 	matrix of the original homogeneous walk, where we assume the same transition probability for every edge in both directions, written below in the way respecting the partition and the  above order of vertices reads
	\begin{equation}
	P_{H}=\frac{1}{3} \left(\begin{array}{c:ccc:ccc:c} 
		0 & 1 & 1 & 1 & 0 & 0 & 0 & 0 \\ \hdashline
		1 & 0 & 0 & 0 & 1 & 1 & 0 & 0 \\ 
		1 & 0 & 0  & 0 & 1 & 0 & 1 & 0 \\ 
		1 & 0 & 0 & 0 & 0 & 1 & 1 & 0\\ \hdashline 
		0 & 1 & 1 & 0 & 0 & 0 & 0 & 1 \\
		0 & 1 & 0 & 1 & 0 & 0 & 0 & 1 \\
		0 & 0 & 1 & 1 & 0 & 0 & 0 & 1 \\ \hdashline
		0 & 0 & 0 & 0 & 1 & 1 & 1 & 0 \\
	\end{array}\right).
\end{equation}
It respects the row sum criterion and gives rise to the lumped Markov chain, see Figure~\ref{fig:Ehrenfests-4}, with the transition matrix
	\begin{equation} \label{eq:red-hexahedron-P}
		\tilde{P}_H =  \frac{1}{3} \begin{pmatrix}
			0 & 3 & 0 & 0 \\
			1 & 0 & 2 & 0 \\
			0 & 2 & 0 & 1 \\
			0 & 0 & 3 & 0
		\end{pmatrix}.
	\end{equation}
	\begin{figure}[h!]
		\begin{center}
			\includegraphics[width=6cm]{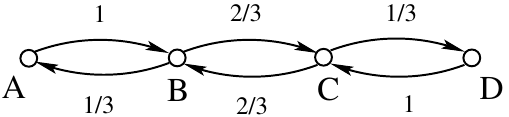}
		\end{center}
		\caption{Lumped random walk on the hexahedron graph}
		\label{fig:Ehrenfests-4}
	\end{figure}
\end{Ex}

\subsection{Szegedy's quantization of Markov chains}
Szegedy's quantization~\cite{Szegedy,Childs-CMP} of a random walk on graph with $|V|=N$ vertices starts with tensor doubling $\CC^N \otimes \CC^N$ --- the first factor is called the position space, and the second factor is called the coin space. The state $|i \rangle \otimes |j \rangle $ will be interpreted as "particle in position $i$ looks at the position $j$". The stochastic matrix $P$ allows to define normalized orthogonal vectors
\begin{equation} \label{eq:phi-i}
	| \phi_i \rangle =  \ket{i}\otimes \sum_{j=1}^N \sqrt{P_{ij}}\,  \ket{j} , \qquad i = 1,\dots , N.
\end{equation} 
By $\Pi$ denote the orthogonal projection on the subspace generated by the vectors  $|\phi_i\rangle$, i.e. 
\begin{equation} \label{eq:Pi}
	\Pi = \sum_{i=1}^N \ket{\phi_i } \bra{ \phi_i}, \qquad \text{i.e.} \quad \Pi (\ket{i}\otimes \ket{j}) = \ket{\phi_{i}}\sqrt{P_{ij}},
\end{equation}
then the reflection
\begin{equation} \label{eq:R}
	R = 2\Pi - \II,
\end{equation}
which acts on the coin space only, will be called the Grover~\cite{Grover} coin flip operator. 
Let $S$ be the operator that swaps the two registers 
\begin{equation} \label{eq:S}
	S( \ket{i}  \otimes \ket{j}  ) =     \ket{j}  \otimes \ket{i},
\end{equation}
then the single step of the quantum walk is defined as the unitary operator being the composition of coin flip and the position swap
\begin{equation} \label{eq:U}
	U = S \, R.
\end{equation}
\begin{Ex} \label{ex:hexahedron-Q}
	The quantization of the random walk on the hexahedron graph starts with definition of eight vectors $\ket{\phi_i}$, $i \in \ZZ_2^3$, of the space $\CC^8 \otimes \CC^8$
	\begin{gather*} \label{eq:phi-hex}
	\ket{\phi_{000}} = \frac{1}{\sqrt{3}} \ket{000}\otimes \left( \ket{001} + \ket{010} + \ket{100} \right), \quad  \dots \;  , \quad \ket{\phi_{111}} \frac{1}{\sqrt{3}} \ket{111}\otimes \left( \ket{110} + \ket{011} + \ket{101} \right).  
	\end{gather*}
An example of nontrivial action of the evolution operator looks as follows
\begin{equation*}
	\ket{001}\otimes \ket{101} \stackrel{R}{\mapsto} \frac{1}{3} 	\ket{001}\otimes \left(  2\ket{000} +2 \ket{011} -  \ket{101} \right) \stackrel{S}{\mapsto} \frac{1}{3} 	 \left(  2\ket{000} +2 \ket{011} -  \ket{101} \right) \otimes \ket{001} .
\end{equation*}
\end{Ex}

Let us  present quantization of the lumped random walk on the hexahedron described in Example~\ref{ex:hexahedron-E}. 
\begin{Ex} \label{ex:hexahedron-r-Q} 
	Consider the quantization of the reduced Markov chain with the transition matrix \eqref{eq:red-hexahedron-P} studied in Example~\ref{ex:hexahedron-E}. The vectors $\ket{\phi_i}$, $i=A, B, C, D$ are given by
	\begin{gather*}
		\ket{\phi_A} = \ket{A}\otimes \ket{B}, \quad  \ket{\phi_B} = \frac{1}{\sqrt{3}}  \ket{B}\otimes \left( \ket{A} + \sqrt{2} \ket{C} \right),\\
		   \ket{\phi_C} = \frac{1}{\sqrt{3}}  \ket{C}\otimes \left( \sqrt{2} \ket{B} +  \ket{D} \right),  \quad  \ket{\phi_D} = \ket{D}\otimes \ket{C} ,
	\end{gather*}
and the action of evolution operator $U$ on the natural basis vectors reads as follows
\begin{align*}
	U (\ket{A}\otimes \ket{B}) = \ket{B}\otimes \ket{A}, \qquad &	
	U (\ket{B}\otimes \ket{A} ) = 
	-\frac{1}{3} \ket{A}\otimes \ket{B} +  \frac{2\sqrt{2}}{3} \ket{C}\otimes \ket{B} , \\ 	U (\ket{B}\otimes  \ket{C} ) = 
	\frac{2\sqrt{2}}{3}  \ket{A}\otimes \ket{B} + \frac{1}{3} \ket{C}\otimes \ket{B} , \qquad &
		U (\ket{C} \otimes \ket{B}) = \frac{1}{3} \ket{B}\otimes \ket{C} + 	\frac{2\sqrt{2}}{3} \ket{D}\otimes \ket{C}, \\
			U (\ket{C} \otimes \ket{D}) =  \frac{2\sqrt{2}}{3} \ket{B}\otimes \ket{C}   - \frac{1}{3}\ket{D}\otimes \ket{C}, \qquad &
				U (\ket{D} \otimes \ket{C}) = \ket{C} \otimes \ket{D}.
\end{align*}	 
\end{Ex}

\subsection{Cantero--Moral--Vel\'{a}zquez (CMV) and Jacobi matrices}
Consider a unitary operator $U$ acting in a Hilbert space $\HH$. Given unit cyclic vector $e_0\in \HH$, i.e. finite linear combinations of $\{ U^n e_0 \}_{n\in\ZZ}$ are dense in $\HH$, define~\cite{CMV-1,CMV-2} the CMV basis $\{e_n\}_{n=0}^\infty$  by orthonormalizing the sequence $e_0, U e_0,U^{-1}e_0, U^2 e_0, U^{-2} e_0, \dots $. 
The matrix $\mathcal{C}$ with entries given by
\begin{equation*}
	\mathcal{C}_{mn} = \langle e_m , U e_n \rangle.
\end{equation*}
is unitary and pentadiagonal 
\begin{equation*}
	\mathcal{C} = \left( \begin{array}{ccccccc}
		\bar{\alpha}_0 & \bar{\alpha}_1 \rho_0 & \rho_1\rho_0 & 0    & 0& 0 & \cdots\\
		\rho_0  & -\bar{\alpha}_1 \alpha_0 & -\rho_1 \alpha_0 & 0& 0& 0 & \cdots \\
		0 & \bar{\alpha}_2 \rho_1 & -\bar{\alpha}_2\alpha_1 & \bar{\alpha}_3 \rho_2 & \rho_3 \rho_2 &0& \cdots \\
		0 & \rho_2 \rho_1 & -\rho_2\alpha_1 & -\bar{\alpha}_3 \alpha_2 & - \rho_3 \alpha_2 & 0&\cdots \\
		0 & 0 & 0 & \bar{\alpha}_4 \rho_3 & -\bar{\alpha}_4\alpha_3 & \bar{\alpha}_5 \rho_4 & \cdots \\
		0 & 0 & 0 & \rho_4 \rho_3 & -\rho_4\alpha_3 & -\bar{\alpha}_5 \alpha_4 & \cdots  \\
		\cdots & \cdots & \cdots & \cdots & \cdots & \cdots & \cdots
	\end{array} \right),
\end{equation*}
where the Verblunsky coefficients $\alpha_0, \alpha_1, \alpha_2,\dots$ satisfy $|\alpha_j| < 1$, and 
\begin{equation*}
\rho_n = \sqrt{1 - |\alpha_n|^2}.
\end{equation*}

The CMV matrix $\mathcal{C}$ has decomposition $\mathcal{C} = \mathcal{L}\mathcal{M}$, with
\begin{equation*}
	\mathcal{L} = \Theta_0 \oplus \Theta_2 \oplus \Theta_4 \oplus \dots \, , \quad
	\mathcal{M} = 1 \oplus \Theta_3 \oplus \Theta_ 5\oplus \dots \, ,  \quad
	\Theta_k = \left( \begin{array}{cc}
		\bar{\alpha}_k & \rho_k \\ \rho_k & - \alpha_k
	\end{array} \right).
\end{equation*}
It turns out that any cyclic unitary model is unitarily equivalent to a unique CMV model $(\ell^2(\NN),\mathcal{C},e_0)$ with $e_0 = (1,0,0, \dots)^T$.

\begin{Rem}
	Finite $N\times N$ CMV matrices are characterized by $|\alpha_k| < 1$, $k=0,\dots ,N-2$, and $|\alpha_{N-1}|=1$. Each unitary equivalence class of $N\times N$ unitary matrices with a cyclic vector contains a unique finite CMV matrix with cyclic vector $e_0 = (1,0,\dots,0)^T$. 
\end{Rem}
\begin{Rem}
	Such matrices appeared earlier in the numerical analysis literature~\cite{Watkins} as an analog for unitary matrices of the Lanczos algorithm~\cite{Lanczos,GolubLoan} of bringing symmetric matrices to the tridiagonal form.
\end{Rem}

For real Verblunsky coefficients one has $\mathcal{M}^2 = \mathcal{L}^2 = \II$. 
The restriction of the self-adjoint operator
\begin{equation}
\frac{1}{2} \left(\mathcal{C} + \mathcal{C}^t \right),
\end{equation}
to the eigenspace $\mathcal{M}$ corresponding to eigenvalue $1$, with $e_0$ as the cyclic vector~\cite{KillipNenciu-CMV,Simon-OPUC}, is of the Jacobi form
\begin{equation*} 
	\mathcal{J} = \left( \begin{array}{ccccc}
		r_0 & s_0 & 0 & 0 & \cdots \\
		s_0 & r_1 & s_1 & 0 & \cdots \\
		0 & s_1 & r_2 & s_2 & \cdots \\
		0 & 0 & s_2 & r_3 & \cdots \\
		\vdots & \vdots & \vdots & \vdots & \ddots
	\end{array} \right),
\end{equation*}
where the Jacobi matrix coefficients $r_k$ and $s_k$, $k=0,1,\dots$, are given by the Geronimus relations
\begin{align}
	s_k & =  \frac{1}{2}\sqrt{ (1-\alpha_{2k-1}) (1 - \alpha_{2k}^2) (1+\alpha_{2k+1})},\\
	\label{eq:rk}
	r_k & = \frac{1}{2} \left( \alpha_{2k}(1-\alpha_{2k-1}) - \alpha_{2k-2}(1+\alpha_{2k-1})\right).
\end{align}

In \cite{Doliwa-Siemaszko-QW} it was shown that application of Szegedy's quantization to generic discrete-time random walks on half-line $\NN_0$ with one-step transition probabilities given by the stochastic matrix, see~Figure~\ref{fig:random-walk},
\begin{equation*}
	P = \left( \begin{array}{ccccc}
		r_0 & p_0 & 0 & 0 & \cdots \\
		q_1 & r_1 & p_1 & 0 & \cdots \\
		0 & q_2 & r_2 & p_2 & \cdots \\
		0 & 0 & q_3 & r_3 & \cdots \\
		\vdots & \vdots & \vdots & \vdots & \ddots
	\end{array} \right),
\end{equation*}
\begin{figure}[!ht]
	\begin{center}
		\includegraphics[width=8cm]{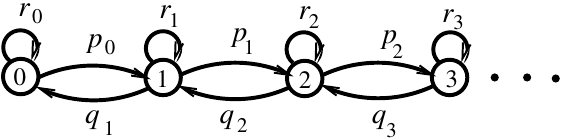}
	\end{center}
	\caption{discrete-time random walk on the half-line}
	\label{fig:random-walk}
\end{figure}
leads to the following result.
\begin{Prop} \label{prop:pqr-a}
	The CMV basis of the quantum evolution operator for Szegedy's quantization of the random walk on the half-line with the cyclic vector $e_0 = | \phi_0 \rangle$ has real Verblunsky coefficients related to the random walk transition probabilities by the formulas \eqref{eq:rk} and
	\begin{align}  \label{eq:qk}
		q_k & = \frac{1}{2} ( 1 + \alpha_{2k-2}) (1 + \alpha_{2k-1}), \\
		\label{eq:pk}
		p_k & = \frac{1}{2} ( 1 - \alpha_{2k-1}) (1 - \alpha_{2k}).
	\end{align}
\end{Prop}
\begin{Rem}
	Recall that we put $\alpha_{-1} = -1$ what gives $q_0 = 0$, $r_0 = \alpha_0$, $p_0 = 1- \alpha_0$.
\end{Rem}
\begin{Cor}
	In \cite{Doliwa-Siemaszko-QW} it was also shown that when the cyclic vector $e_0$ is invariant with respect to the action of the coin operator $R$ in the Szegedy quantization then the corresponding CMV basis and Verblunsky coefficients can be found from the recurrence
	\begin{equation}
		S(e_{2k}) = \alpha_{2k} e_{2k} + \rho_{2k} e_{2k+1}, \qquad R(e_{2k+1}) = \alpha_{2k+1} e_{2k+1} + \rho_{2k+1} e_{2k+2}, \qquad k\geq 0. 
	\end{equation}
\end{Cor}
\begin{Ex} \label{ex:hexahedron-r-Q-V}
	For the quantized Markov chain with the transition matrix~\eqref{eq:red-hexahedron-P} obtained by lumping the random walk on the hexahedron graph, as presented in Examples~\ref{ex:hexahedron-E} and~\ref{ex:hexahedron-r-Q}, and with the cyclic vector $e_0 = | \phi_A \rangle$, the corresponding CMV basis vectors and the Verblunsky coefficients read as follows
	\begin{gather*}
		e_0 = \ket{A}\otimes \ket{B}, \qquad \alpha_0 = 0, \qquad e_1 = \ket{B} \otimes \ket{A}, \qquad \alpha_1 = -\frac{1}{3}, \qquad e_2 = \ket{B}\otimes \ket{C}, \qquad \alpha_2 = 0, \\ e_3 =\ket{C}\otimes \ket{B}, \qquad 
		\alpha_3 = \frac{1}{3}, \qquad e_4 = \ket{C}\otimes \ket{D}, \qquad \alpha_4 = 0, \qquad e_5 = \ket{D}\otimes \ket{C}, \qquad \alpha_5 = 1.
	\end{gather*}
\end{Ex}
The next Example demonstrates that choosing an appropriate initial vector, which does not have to be cyclic vector for the whole space, in the CMV procedure may result in a quantum walk in a reduced space. The correspondence between the two quantum walks suggests general aggregation procedure which we develop in the next Section.
\begin{Ex}
	Consider quantum walk operator on the hexahedron graph constructed in Example~\ref{ex:hexahedron-Q} and start the CMV procedure with the initial vector $e_0 = \ket{\phi_{000}}$. With the operators $R$ defined by the vectors $\ket{\phi_i}$, $i\in\BB^3$, given in \eqref{eq:phi-hex} we obtain \emph{the same} Verblunsky coefficients as in the above Example~\ref{ex:hexahedron-r-Q-V} and the following CMV basis
	\begin{gather*}
		e_0 = \frac{1}{\sqrt{3}}\ket{000}\otimes\left( \ket{001} + \ket{010} + \ket{100} \right), \qquad  e_1 = \frac{1}{\sqrt{3}}\left( \ket{001} + \ket{010} + \ket{100} \right) \otimes \ket{000}, \\
		 e_2 = \frac{1}{\sqrt{6}} \left( \ket{001}\otimes \ket{011}
+ \ket{001}\otimes \ket{101} + \ket{010}\otimes \ket{011} + \ket{010}\otimes \ket{110}	+
\ket{100}\otimes \ket{101}	+ \ket{100}\otimes \ket{110} \right) , \\ 
		 e_3 =\frac{1}{\sqrt{6}} \left( \ket{011}\otimes \ket{001}
		 + \ket{011}\otimes \ket{010} + \ket{101}\otimes \ket{001} + \ket{101}\otimes \ket{100}	+
		 \ket{110}\otimes \ket{010}	+ \ket{110}\otimes \ket{100} \right) , \\ 
		 e_4 = \frac{1}{\sqrt{3}}\left( \ket{011} + \ket{101} + \ket{110} \right) \otimes \ket{111}, \qquad e_5 = \frac{1}{\sqrt{3}}\ket{111}\otimes\left( \ket{011} + \ket{101} + \ket{110} \right).
	\end{gather*}
\end{Ex}

\begin{Rem}
The structure of the CMV representation implies that any quantum walk with  a time independent evolution operator  can be mapped to a corresponding walk on a weighted path.
\end{Rem}

\section{Quantum aggregation} \label{sec:qa}
\subsection{Formulation of the problem}
Consider Markov chain with the state set $V$ and transition matrix $P$, the corresponding quantum walk is constructed by equations~\eqref{eq:phi-i}-\eqref{eq:U}. Given partition $\tilde{V}$ of $V$, construct the corresponding coarse-grained graph with the edge set $\tilde{E}$, as described in equation~\eqref{eq:tilde-E}. Consider a random walk on the graph $\tilde{G} = (\tilde{V}, \tilde{E})$ with a transition matrix $\tilde{P}$ --- we do not specify at the moment a relation between the both transition matrices. Our goal is to investigate the aggregation procedure on the quantum level. 

Denote by $U$ and $\tilde{U}$ the corresponding evolution operators of the quantized random walks. Inspired by Examples presented in the previous Section, define the aggregated states
\begin{equation} \label{eq:aggregated-uv}
	\ket{u,v} = \sum_{\substack{i \in u, j \in v \\
	(i,j)\in E}} a_{ij} \ket{i} \otimes \ket{j}, \qquad (u,v)\in \tilde{E},
\end{equation}
where properties of the matrix $(a_{ij})_{i,j\in V}$  will be determined from the condition that the action of the evolution operator $U$ on $\ket{u,v}$ must be the same as the action of the evolution operator $\tilde{U}$ on $\ket{u}\otimes \ket{v}$. By definition $a_{ij} = 0$ for $(i,j)\notin E$. This condition is automatically satisfied if we introduce $q_{ij}$ such $a_{ij} = q_{ij} \sqrt{P_{ij}}$. It will be convenient to allow $q_{ij}\neq 0$, even if $P_{ij}=0$.
\begin{Rem}
	Notice that the above condition can be applied to generic quantum walks on graphs. 	We do not study the quantum aggregation problem in full generality, but we restrict our attention to the Szegedy quantized random walks and real coefficients $a_{ij}$ only.
\end{Rem}
Because the evolution operator is defined as superposition of the reflection and swap operators, we demand that the following two conditions are satisfied separately assuming the correspondence
	\begin{equation}
	\ket{u,v} \leftrightarrow |u\rangle \otimes |v \rangle. 
\end{equation}
1. The action of the projection operator $\Pi$, in the Hilbert space of the quantum walk on $G$, reduces properly, compare with equation~\eqref{eq:Pi}, 
	\begin{equation} \label{eq:red-Pi}
		\Pi ( \ket{u,v} )  = |\phi_u \rangle \sqrt{\tilde{P}_{uv}}, \qquad \text{where} \qquad  |\phi_u \rangle = \sum_{v\in \tilde{V}} \sqrt{\tilde{P}_{uv}} \ket{u,v} .
	\end{equation}	
2. The action of the swap operator $S$, in the Hilbert space of the quantum walk on $G$, reduces properly  
	\begin{equation} \label{eq:red-S}
		S(\ket{u,v} ) = \ket{v,u} .
	\end{equation}
\begin{Ex}
	In the reduction of quantum walk on the hexahedron graph considered in Examples~\ref{ex:hexahedron-r-Q} and \ref{ex:hexahedron-r-Q-V} the correspondence follows from identification of the vectors $e_i$, $i = 0, \dots , 5$, and gives
		\begin{gather*}
			\ket{A,B} = \frac{1}{\sqrt{3}} | 000\rangle \otimes \left(  | 001\rangle +  | 010\rangle +  | 100\rangle \right), 	\qquad
			\ket{B,A} = \frac{1}{\sqrt{3}}  \left(  | 001\rangle +  | 010\rangle +  | 100\rangle \right) \otimes | 000\rangle , \\
			\ket{B,C} = \frac{1}{\sqrt{6}} \left( | 001\rangle \otimes | 011\rangle  + | 001\rangle \otimes | 101\rangle +  | 010\rangle \otimes | 011\rangle +  | 010\rangle \otimes | 110\rangle +  | 100\rangle \otimes | 101\rangle +  | 100\rangle \otimes | 110\rangle\right), \\	
			\ket{C,B} = \frac{1}{\sqrt{6}} \left( | 011\rangle \otimes | 001\rangle  + | 011\rangle \otimes | 010\rangle +  | 101\rangle \otimes | 001\rangle +  |  101 \rangle \otimes | 100\rangle +  | 110\rangle \otimes | 010\rangle +  | 110\rangle \otimes | 100\rangle\right), \\	
			\ket{C,D} = \frac{1}{\sqrt{3}}  \left(  | 011\rangle +  | 101\rangle +  | 110\rangle \right) \otimes | 111\rangle , \qquad	
			\ket{D,C} = \frac{1}{\sqrt{3}} | 111\rangle \otimes \left(  | 011\rangle +  | 101\rangle +  | 110\rangle \right). 	
		\end{gather*}
One can verify directly that the reduction conditions are satisfied. Notice however that the identification does not hold in all the aspects. For example, unlike the state vector $|B\rangle \otimes |C \rangle $, the corresponding vector $\ket{B,C}$ of the aggregated model  is not a product state. The entanglement of the walker with coin can be checked by finding the corresponding density matrix reduced with respect to the coin space (the second factor)
		\begin{equation}
			\mathrm{tr}_C  \ket{B,C}  \bra{B,C} = \frac{1}{6} \begin{pmatrix}
				2 & 1 & 1 \\
				1 & 2 & 1 \\
				1 & 1 & 2
			\end{pmatrix},
		\end{equation}
		which has eigenvalues $\frac{2}{3}, \frac{1}{6}, \frac{1}{6}$. The corresponding von Neumann entropy reads
		\begin{equation}
			S = - \left(\frac{1}{6} \log \frac{1}{6} + \frac{1}{6} \log \frac{1}{6} +\frac{2}{3} \log \frac{2}{3} \right)  = \log 3 - \frac{1}{3} .
		\end{equation}	
Similar remark applies also to the state vector 
		\begin{equation}
			\ket{\phi_B} = \frac{1}{\sqrt{3}}	\ket{B,A} + \frac{\sqrt{2}}{\sqrt{3}}	\ket{B,C},
		\end{equation}
		whose reduced density matrix reads
		\begin{equation}
			\mathrm{tr}_C  \ket{\phi_B} \bra{\phi_B} = \frac{1}{9} \begin{pmatrix}
				3 & 2 & 2 \\
				2 & 3 & 2 \\
				2 & 2 & 3
			\end{pmatrix},
		\end{equation}
		and has eigenvalues $\frac{7}{9}, \frac{1}{9}, \frac{1}{9}$. 
\end{Ex}

\subsection{Solution of the problem}
Before stating a proposition, which is the main result of the paper, let us present a reasoning, which should help in digesting its proof. Let us analyze first the condition \eqref{eq:red-Pi}. By expressing equation~\eqref{eq:Pi} in the basis $\ket{i}\otimes \ket{j}$, $i,j \in V$,  we obtain
\begin{equation}
\sqrt{ P_{ij}}\sum_{k\in v} q_{ik}P_{ik} =	q_{ij} \sqrt{ P_{ij}\tilde{P}_{uv} \tilde{P}_{uw}} , \qquad i \in u, \quad j \in w.
\end{equation}
which is satisfied trivially when $P_{ij}=0$. When $P_{ij}\neq 0$ then the above equation implies that $q_{ij}$ may be chosen in such a way that it does not depend on  a particular index $j\in w$. This allows to make the replacement $q_{ij} = s_{iw}$, and gives equation
\begin{equation} 
s_{iv} \sum_{k\in v} P_{ik} = s_{iw} \sqrt{\tilde{P}_{uv} \tilde{P}_{uw}} , \qquad i\in u.
\end{equation}
Notice that for $w=v$ and under assumption $s_{iv}\neq 0$ this gives (for $s_{iv}=0$ the above equation is trivially satisfied) the strong lumpability condition \eqref{eq:lumpability}. 

If so, that again under assumption $\tilde{P}_{uv} \neq 0$, we obtain the following relation
\begin{equation} \label{eq:sPuv}
	s_{iv} \sqrt{\tilde{P}_{uv}}  = s_{iw} \sqrt{ \tilde{P}_{uw}} , \qquad i\in u.
\end{equation}
The above system of equations is self-consistent in the following sense. We can calculate $s_{iw}$ directly from $s_{iv}$ and probabilities $\tilde{P}_{uv}$ and $\tilde{P}_{uw}$. Another option is to use an intermediate coefficients $s_{iz}$ and corresponding probabilities. It turns out that both paths give the same result
\begin{equation*}
	s_{iw}^2 = s_{iz}^2 \frac{\tilde{P}_{uz}}{\tilde{P}_{uw}} = 
	s_{iv}^2 \frac{\tilde{P}_{uv}}{\tilde{P}_{uz}} \frac{\tilde{P}_{uz}}{\tilde{P}_{uw}} = s_{iv}^2 \frac{\tilde{P}_{uv}}{\tilde{P}_{uw}}.
\end{equation*}
\begin{figure}[!ht]
	\begin{center}
		\includegraphics[width=7cm]{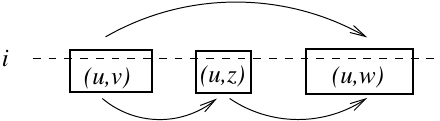}
	\end{center}
	\caption{Self-consistency of equations \eqref{eq:sPuv}}
	\label{fig:PP}
\end{figure}

Condition \eqref{eq:red-S} implies symmetry $a_{ij} = a_{ji}$ of the matrix, and in terms of the linking coefficients $s_{iv}$ and the transition probabilities reads
\begin{equation} \label{eq:sPij}
	s_{iv} \sqrt{P_{ij}}  = s_{ju} \sqrt{P_{ji}} , \qquad i\in u, \quad j\in v.
\end{equation}
It implies, in particular, that for non-trivial linking coefficients
\begin{equation} \label{eq:PP}
	P_{ij}>0 \Leftrightarrow P_{ji} > 0, \qquad i,j \in V.
\end{equation}

\begin{Rem}
	Notice that requirement \eqref{eq:sPij} has the form similar to that of detailed balance condition~\cite{Norris} known from the theory of reversible Markov chains, which is equivalent to certain nonlinear relation (Kolmogorov's cycle criterion~\cite{Kelly}) between elements of the transition matrix. 
\end{Rem}		

Moreover, assume that for $i_1, i_2 \in u$ and $j_1, j_2 \in v$ all the corresponding pairs of probabilities in $(u,v)$ block are non-zero.
\begin{figure}[!ht]
	\begin{center}
		\includegraphics[width=6cm]{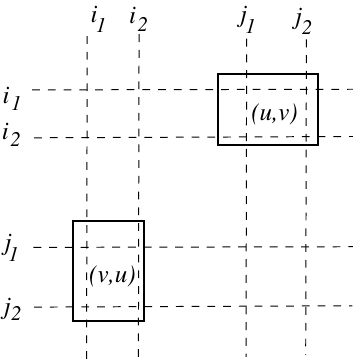}
	\end{center}
	\caption{Self-consistency of equations \eqref{eq:sPij}}
	\label{fig:PPPP}
\end{figure}
Given $s_{i_1 v}$, then $s_{i_2 v}$ can be find in two ways: either using $s_{j_1 u}$
\begin{equation*}
	s_{i_2 v}^2 = 
		s_{j_1 u}^2 \frac{P_{j_1 i_2}}{P_{i_2 j_1}} = s_{i_1 v}^2 \frac{P_{i_1 j_1}}{P_{j_1 i_1}} \frac{P_{j_1 i_2}}{P_{i_2 j_1}} , 
\end{equation*}
either $s_{j_2 u}$
\begin{equation*}
	s_{i_2 v}^2 = 
			s_{j_2 u}^2 \frac{P_{j_2 i_2}}{P_{i_2 j_2}} = s_{i_1 v}^2 \frac{P_{i_1 j_2}}{P_{j_2 i_1}} \frac{P_{j_2 i_2}}{P_{i_2 j_2}} .
\end{equation*}
Since both ways should give the same result we obtain the following constraint on the transition probabilities 
\begin{equation}
	\label{eq:PPPP}
	P_{i_1 j_1} P_{j_1 i_2} P_{i_2 j_2} P_{j_2 i_1} =
P_{i_1 j_2 } P_{ j_2 i_2} P_{i_2 j_1 } 	P_{j_1 i_1 }  .	
\end{equation}

Let us consider mutual compatibility of both equations  \eqref{eq:sPuv} and \eqref{eq:sPij}. Given $i\in u$, $j\in v$ and $k\in w$, assume that probabilities $P_{ij}$, $P_{ik}$ and $P_{jk}$ do not vanish. Starting from $s_{iv}$ one can obtain $s_{kv}$ in two different ways:
\begin{figure}[!ht]
\begin{center}
	\includegraphics[width=7cm]{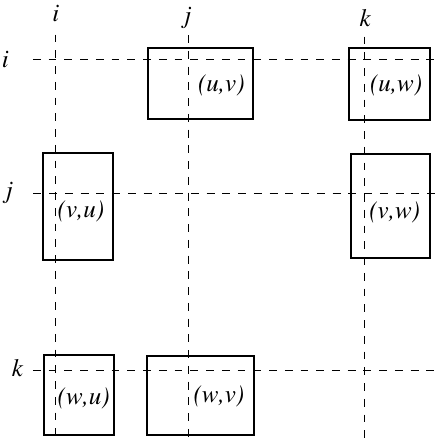}
\end{center}
\caption{Consistency of equations  \eqref{eq:sPuv} and \eqref{eq:sPij}}
\label{fig:PPP}
\end{figure}
\begin{align*}
	I: \qquad & s_{iv} \stackrel{\eqref{eq:sPij}}{\rightsquigarrow} s_{ju} \stackrel{\eqref{eq:sPuv}}{\rightsquigarrow} s_{jw} 
	\stackrel{\eqref{eq:sPij}}{\rightsquigarrow} s_{kv}, \\
	II: \qquad & s_{iv} \stackrel{\eqref{eq:sPuv}}{\rightsquigarrow} s_{iw} \stackrel{\eqref{eq:sPij}}{\rightsquigarrow} s_{ku} 
	\stackrel{\eqref{eq:sPuv}}{\rightsquigarrow} s_{kv} ,
\end{align*}
what implies
\begin{equation*}
	s_{kv}^2 = s_{iv}^2 \frac{P_{ij}}{P_{ji}} \frac{\tilde{P}_{vu}}{\tilde{P}_{vw}} \frac{P_{jk}}{P_{kj}} = 
		s_{iv}^2  \frac{\tilde{P}_{uv}}{\tilde{P}_{uw}} \frac{P_{ik}}{P_{ki}} \frac{\tilde{P}_{wu}}{\tilde{P}_{wv}},
\end{equation*}
or
\begin{equation} \label{eq:PPP}
\frac{\tilde{P}_{uv} \tilde{P}_{vw}\tilde{P}_{wu}}
{\tilde{P}_{uw} \tilde{P}_{wv}\tilde{P}_{vu}} =
\frac{P_{ij} P_{jk} P_{ki}}{P_{ik} P_{kj} P_{ji}} .
\end{equation}
\begin{Rem}
	Notice that for $u=v=w$ and $\tilde{P}_{uu} >0$ equation \eqref{eq:PPP} gives
	Kolmogorov's cyclic condition
	\begin{equation*}
		P_{ij} P_{jk} P_{ki}  = P_{ik} P_{kj} P_{ji},
	\end{equation*}	
	within each single part of the partition. The condition can be also obtained as compatibility of subsystem \eqref{eq:sPij} for $u=v$.
\end{Rem}

The above form of the Kolmogorov condition is sufficient when all cycles in the graph can be obtained from cycles of length three. In the general case it should be generalized for all cycles; it is enough to check it for independent ones only. The complete cycle condition for equations~\eqref{eq:sPuv} and \eqref{eq:sPij} can be formulated as follows: 
let $
i_0,i_1,\ldots,i_m=i_0
$
be a cycle in the support graph $G$ of the Markov chain $P$, then
\begin{equation} 
	\label{eq:cycle-condition}
	\prod_{r=0}^{m-1}
	\frac{P_{i_ri_{r+1}}}{P_{i_{r+1}i_r}}
	=
	\prod_{r=0}^{m-1}
	\frac{\tilde{P}_{u(i_r)u(i_{r+1})}}
	{\tilde{P}_{u(i_{r+1})u(i_r)}},
\end{equation}	
where the vertex $i$ belongs to the cluster $u(i)$.
Equation \eqref{eq:PPP} is a special case of \eqref{eq:cycle-condition} for cycles of length three. Equation \eqref{eq:PPPP} corresponds to certain cycles of length four, in which the vertices alternate between two fixed partition classes.

Finally notice, that the normalization condition of the state vector
\begin{equation} \label{eq:aggr-state}
	\ket{u,v} = \sum_{i\in u, j\in v} s_{iv}\sqrt{P_{ij}} \ket{i}\otimes \ket{j}
\end{equation} 
together with lumpability \eqref{eq:lumpability} imply that, whenever $\tilde{P}_{uv} > 0$,
\begin{equation} \label{eq:s-norm}
\tilde{P}_{uv} \sum_{i\in u} s_{iv}^2 = 1.
\end{equation}
The above condition is compatible with equation~\eqref{eq:sPuv}, as both directly lead to
\begin{equation*}
	\tilde{P}_{uw} \sum_{i\in u} s_{iw}^2 = 1.
\end{equation*}
Moreover, equations \eqref{eq:s-norm} and \eqref{eq:sPij} together with the lumpability condition \eqref{eq:lumpability} imply
\begin{equation*}
\tilde{P}_{vu} \sum_{j\in v} s_{ju}^2 = \sum_{j\in v} s_{ju}^2 \sum_{i\in u} P_{ji} = \sum_{i\in u, j\in v} s_{iv}^2 P_{ij} = 1,
\end{equation*}
what gives consistency of the normalization condition with \eqref{eq:sPij}. 

Above, we derived the (necessary) conditions resulting from the permutability of the aggregation--quantization operations. Below, we demonstrate that these are, in fact, sufficient conditions as well.
\begin{Prop} \label{prop:P}
Given lumpable Markov chain with respect to the partition $\tilde{V}$ of its state space $V$, which satisfies the weak reversability condition \eqref{eq:PP} and the consistency conditions  \eqref{eq:cycle-condition}. When the corresponding graph is connected
then there exists unique aggregation of Szegedy's quantization of the original chain, with non-negative linking coefficients, isomorphic to quantization of the lumped chain.
\end{Prop}
\begin{proof}
Our goal is to construct vectors $s_v = (s_{iv})_{i\in V}$, $v\in \tilde{V}$ of linking coefficients, which satisfy constraints~\eqref{eq:sPuv}, \eqref{eq:sPij} and \eqref{eq:s-norm}. They allow to construct aggregated quantum states \eqref{eq:aggr-state}. Then, by reversing the above reasoning we infer that the proper reduction conditions \eqref{eq:red-Pi} and \eqref{eq:red-S} will be satisfied.	
	
Consider $u,v\in \tilde{V}$ such that $\tilde{P}_{uv}>0$, and choose $i\in u$ and $j\in v$ with $P_{i j}>0$. Declare 
$s_{i v} = a$, and by compatible equations~\eqref{eq:sPuv} and \eqref{eq:sPij}, which allow to move along the row changing the cluster and to change the row, respectively, define the linking coefficients. By the graph connectivity and compatibility of the moves this process gives all the coefficients dependent homogeneously on the parameter $a$. Finally, the normalization condition~\eqref{eq:s-norm} fixes the parameter up to a sign.
\end{proof}

\begin{Rem}
	The procedure of obtaining of the nonlinear  conditions  \eqref{eq:cycle-condition} for lumpable Markov chain as consistency of the linear equations~\eqref{eq:sPuv}-\eqref{eq:sPij} resembles the fundamental ideas of the theory of integrable systems~\cite{IDS}.  It can be also compared with Kolmogorov's cycle condition for transition probabilities of reversible Markov chains~\cite{Kelly} being equivalent to the detailed balance condition.
\end{Rem}
\begin{Rem}
	The graph connectivity condition is not essential, as the procedure can be performed for each its connected component separately.
\end{Rem}

	\begin{Cor}
		For reversible Markov chain $P$ the left hand side of equation \eqref{eq:cycle-condition} is equal to one. By the known result \cite{KemenySnell} that strong lumpability preserves reversibility we obtain that every reversible, strongly lumpable Markov chain satisfies the cycle compatibility condition \eqref{eq:cycle-condition}, and then the proposed quantum aggregation exists.
	\end{Cor}
\begin{Cor} \label{cor:hom}
	A natural example of Markov chains satisfying assumptions of Proposition~\ref{prop:P} is provided by random walks on graphs equipped with equitable partitions~\cite{GodsilRoyle}, which are characterized that the number
	of neighbors in $v$ of a vertex $i\in u$, denoted by $d_{uv}$, depends only on the choice of $u$ and $v$. Then every vertex in $u$ has the same valency
	\begin{equation*}
		d_u = \sum_{v\in \tilde{V}} d_{uv},
	\end{equation*}
and the natural random walk with transition probability $1/d_u$ from vertex $i\in u$ to its neighbors is lumpable. The lumped transition matrix and the aggregated quantum state read 
\begin{equation}
\tilde{P}_{uv} = \frac{d_{uv}}{d_u} , \qquad 	\ket{u,v} = \frac{1}{d_{uv} |u|} \sum_{\substack{i \in u, j \in v \\ (i,j) \in E}} \ket{i}\otimes \ket{j}.
\end{equation} 
To conclude it is enough to observe that the number of edges connecting vertices of the clusters $u$ and $v$ equals
\begin{equation*}
	d_{uv} |u| = d_{vu} |v|.
\end{equation*}
\end{Cor}
\begin{Rem}
Partitions with the above property are provided, for example, as orbits of a subgroup of the automorphism group of a  graph. 
Perfect state transfer via continuous-time quantum walk on quotient graphs obtained from equitable partitions was studied in \cite{Salimi,Ge-Tamon,Godsil-transfer,Bachman-Tollefson,KemptonTolbert}. As it was mentioned in the Introduction, quantum walk based search algorithms in relation to equitable partitions were considered in~\cite{Ide}.
\end{Rem}

\begin{Rem}
	Assumptions of Corollary \ref{cor:hom} are satisfied in the case of distance-regular graphs~\cite{BCN,DevroyeSbihi,DamKoolenTanaka} and their natural partitions into slices/spheres  consisting of points with constant distance from a fixed vertex of such a graph. 
\end{Rem}
Recall, that a connected simple graph is called distance-regular if for any two vertices $x$ and $y$, the number of vertices at distance $j$ from $x$ and at distance $k$ from $y$ depends only upon $j$, $k$, and the distance between $x$ and $y$. This simple definition forces a great deal of structure upon such graphs. In connection to quantum computation we remark that entanglement of free fermions on distinguished distance-regular graphs, such as the Hadamard, Hamming or Johnson graphs, was the subject of recent works~\cite{BernardCrampeVinet-H,BernardCrampeVinet-J,CrampeGuoVinet-H}. 

\begin{Ex}
	Let us present an example, which goes beyond space uniform walks on graphs equipped with equitable partitions. Notice that the following partition
	\begin{equation*}
		K = \{ 001, 010\}, \qquad L = \{ 000, 011, 110, 101\}, \qquad M = \{ 100, 111\},
	\end{equation*}
	of the vertex set of the hexahedron graph (see Example~\ref{ex:hexahedron-E}) is not equitable; for example vertex $000$ of $L$ is connected to both vertices of $K$, while vertex $110$ is connected to $010$ and not to $001$. Let us consider random walk with transition probabilities given by 
	\begin{equation} \label{eq:hexahedron-P-2}
		P=\frac{1}{4} \left(\begin{array}{cc:cccc:cc} 
			0 & 0 & 1 & 1 & 0 & 2 & 0 & 0 \\ 
			0 & 0 & 1 & 1 & 2 & 0 & 0 & 0 \\ \hdashline
			1 & 1 & 0  & 0 & 0 & 0 & 2 & 0 \\ 
			1 & 1 & 0 & 0 & 0 & 0 & 0 & 2\\ 
			0 & 2 & 0 & 0 & 0 & 0 & 1 & 1 \\
			2 & 0 & 0 & 0 & 0 & 0 & 1 & 1 \\ \hdashline
			0 & 0 & 2 & 0 & 1 & 1 & 0 & 0 \\ 
			0 & 0 & 0 & 2 & 1 & 1 & 0 & 0 \\
		\end{array}\right),
	\end{equation}
	with the  above order of vertices. One can check that it satisfies condition~\eqref{eq:lumpability} and gives rise to the lumped Markov chain with the transition matrix
	\begin{equation} \label{eq:red-hexahedron-P-2}
		\tilde{P} =  \frac{1}{2} \begin{pmatrix}
			0 & 2 & 0 \\
			1 & 0 & 1 \\
			0 & 2 & 0 
		\end{pmatrix}.
	\end{equation}
	The corresponding quantum evolution operator act in the natural basis vectors as follows
	\begin{equation} \label{eq:evol-red-hexahedron-P-2} \begin{split}
			\tilde{U} (\ket{K}\otimes \ket{L}) = \ket{L}\otimes \ket{K}, \qquad &	
			\tilde{U} (\ket{L}\otimes \ket{K} ) =  \ket{M}\otimes \ket{L} , \\ 	\tilde{U} (\ket{L}\otimes \ket{M}) = \ket{K}\otimes \ket{L}, \qquad &	
			\tilde{U} (\ket{M}\otimes \ket{L} ) =  \ket{L}\otimes \ket{M} .
		\end{split}
	\end{equation}
	
	Consider the aggregated basis of the reduced Hilbert space, constructed according to \eqref{eq:aggr-state}
	\begin{gather*}
		\ket{K,L} = \frac{1}{2\sqrt{2}} \left( | 001\rangle +  | 010\rangle \right) \otimes \left(  | 000\rangle +  | 011\rangle \right) + \frac{1}{2} \left( | 001\rangle  \otimes | 101\rangle  +  | 010\rangle \otimes  | 110\rangle \right) 	
		, 	\qquad
		\ket{L,K} = S ( 	\ket{K,L}) , \\
		\ket{L,M} = \frac{1}{2\sqrt{2}} \left( | 110\rangle +  | 101\rangle \right) \otimes \left(  | 100\rangle +  | 111\rangle \right) + \frac{1}{2} \left( | 000\rangle  \otimes | 100\rangle  +  | 011\rangle \otimes  | 111\rangle \right) 	
		, 	\qquad
		\ket{M,L} = S ( 	\ket{L,M}), 	
	\end{gather*}	
	where $S$ is the swap operator of the full Hilbert space of Szegedy's quantization of the random walk with transition matrix \eqref{eq:hexahedron-P-2}. The corresponding aggregated version of equations \eqref{eq:evol-red-hexahedron-P-2} can be obtained directly by using the quantum evolution operator of the full Hilbert space. For example, in accordance with \eqref{eq:red-Pi}, we have
	\begin{equation*}
		\Pi (	\ket{L,M}) =  \frac{1}{2\sqrt{2}} \left(\ket{\phi_{110}} + \ket{\phi_{101}}+ \ket{\phi_{000}} + \ket{\phi_{011}} \right) = \frac{1}{2} \left( 	\ket{L,M} + 	\ket{L,K}\right) = \frac{1}{\sqrt{2}} \ket{\phi_L},
	\end{equation*}
	where
	\begin{gather*}
		\ket{\phi_{110}} = \frac{1}{2} \ket{110} \otimes \left( \sqrt{2} \ket{010} + \ket{100} + \ket{111} \right) , \qquad
		\ket{\phi_{101}} = \frac{1}{2} \ket{101} \otimes \left( \sqrt{2} \ket{001} + \ket{100} + \ket{111} \right) ,\\
		\ket{\phi_{000}} = \frac{1}{2} \ket{000} \otimes \left( \sqrt{2} \ket{100} + \ket{001} + \ket{010} \right) , \qquad
		\ket{\phi_{011}} = \frac{1}{2} \ket{011} \otimes \left( \sqrt{2} \ket{111} + \ket{001} + \ket{010} \right) ,
	\end{gather*}
	what gives 
	\begin{equation*}
		U(	\ket{L,M}) = S (2\Pi - \II) (	\ket{L,M}) = 	\ket{K,L}.
	\end{equation*}
\end{Ex}
\begin{Rem}
	Notice that the transition matrix \eqref{eq:red-hexahedron-P-2} can be also obtained from a uniform random walk on a square graph  equipped with the partition $\{00\}$, $\{ 01, 10\}$, $\{11\}$.
\end{Rem}
\section{Quantum walks on Platonic solids and their reductions}
Let us present examples of the aggregation of quantum walks. Because the possibility of strong lumping is related in certain degree to symmetries of the graph and of the corresponding partition, we will be studying graph exhibiting  symmetries. Inspired by the aggregation of the quantum walk on hexahedron, in this Section we briefly present analogous results for other four Platonic solids: tetrahedron $T$, octahedron $O$, icosahedron $I$, and dodecahedron $D$, for completeness~\cite{Coxeter}. In each case we consider homogeneous probabilities, and the partition consisting with the sets of vertices with a fixed distance from one of them. It turns out that in all such cases the lumped Markov chain represents random walk on a path graph. 

We do not write down transition matrices of the initial homogeneous random walks and the basis of the reduced space, as they can be readily obtained from the presented graphs by Corollary~\ref{cor:hom}. In each case we present transition matrices of lumped chains and action of the evolution operators in basis of the reduced space, the Verblunsky coefficients and the corresponding CMV basis.

\subsection{Tetrahedron}
The tetrahedron graph, visualized on Figure~\ref{fig:tetrahedron}, is the complete graph with four vertices enumerated $V(T) = \{ 0,1,2,3\}$. Consider the partition consisting of the subsets $A=\{ 0\}$, $B=\{ 1,2,3\}$.
	\begin{figure}[h!]
		\begin{center}
			\includegraphics[width=4.5cm]{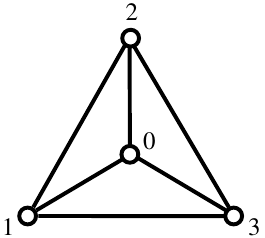}
		\end{center}
		\caption{The tetrahedron graph}
		\label{fig:tetrahedron}
	\end{figure}
The transition matrix of the homogeneous random walks on the graph indexed according to the above order reads
	\begin{equation}
		P_T=\frac{1}{3} \left(\begin{array}{c:ccc} 
			0 & 1 & 1 & 1 \\ \hdashline
			1 & 0 & 1 & 1\\ 
			1 & 1 & 0  & 1\\ 
			1 & 1 & 1 & 0 
		\end{array}\right).
	\end{equation}
The lumping procedure leads to the Markov chain with the transition matrix 
	\begin{equation} \label{eq:P-lT}
		\tilde{P}_T =\frac{1}{3} \left(\begin{array}{cc} 
			0 & 3 \\ 
			1 & 2
		\end{array}\right).
	\end{equation}

The vectors which span the reduced space are given as follows
	\begin{gather*}
		\ket{{A,B}} = \frac{1}{\sqrt{3}} \ket{0} \otimes \left(\ket{1}+\ket{2} + \ket{3}\right), 	\qquad
		\ket{{B,A}} = \frac{1}{\sqrt{3}}  \left(\ket{1}+\ket{2} + \ket{3}\right) \otimes \ket{0} , \\
		\ket{{B,B}} = \frac{1}{\sqrt{6}} \left(\ket{1} \otimes \ket{2} + \ket{1} \otimes \ket{3} + \ket{2} \otimes \ket{1} + \ket{2} \otimes \ket{3} +  \ket{3} \otimes \ket{1} + \ket{3} \otimes \ket{2}\right).
	\end{gather*}
The action of the evolution operator $U$ of the full space restricts to that vectors and reads
	\begin{gather*}
		U(\ket{{A,B}}) = \ket{{B,A}}, \qquad	U(\ket{{B,A}} )= -\frac{1}{3}\ket{{A,B}}+\frac{2\sqrt{2}}{3}\ket{{B,B}}, \qquad U \ket{{B,B}} = \frac{2\sqrt{2}}{3}\ket{{A,B}}+\frac{1}{3}\ket{{B,B}},
	\end{gather*}
in agreement with the quantization of the lumped chain. 

The Verblunsky coefficients and the CMV basis of the quantization of the random walk with the initial state $	e_0 = \ket{\phi_0} = \ket{A,B}$
read
	\begin{equation}
		\alpha_0 = 0, \quad e_1=\ket{B,A}, \qquad 
		\alpha_1 = -\frac{1}{3}, \quad e_2=\ket{B,B}, \qquad \alpha_2=1.
	\end{equation}
\begin{Rem}
The same Verblunsky coefficients are obtained for quantization of the lumped walk \eqref{eq:P-lT} and the cyclic initial vector 
$ \ket{A}\otimes\ket{B}$.
\end{Rem}

\subsection{Octahedron}
The six vertices of the octahedron graph, visualized on Figure~\ref{fig:octahedron}, can be conveniently described as $V(O) = \{\pm1, \pm2, \pm3\} $. 
	\begin{figure}[h!]
	\begin{center}
		\includegraphics[width=4.5cm]{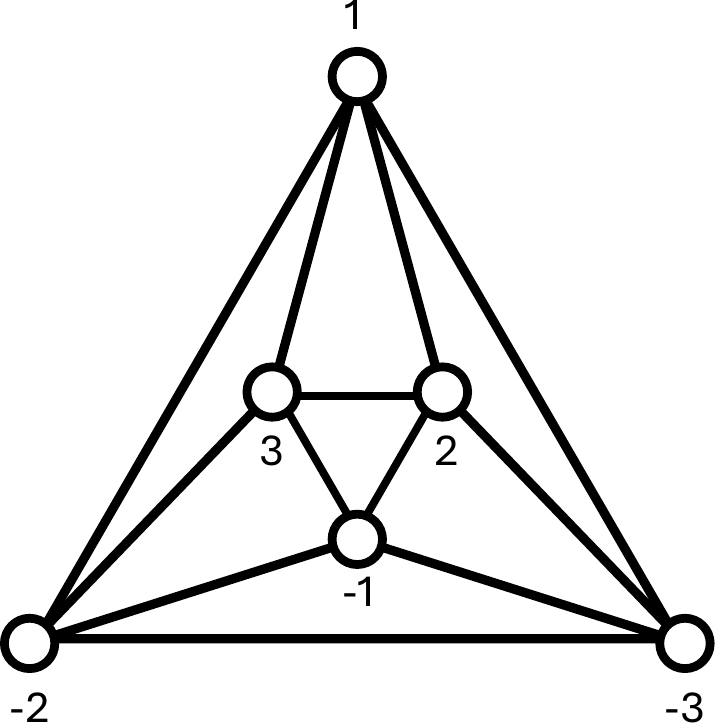}
	\end{center}
	\caption{The octahedron graph}
	\label{fig:octahedron}
\end{figure}
The strong lumpability with respect to the partition $A=\{-1\}$, $B=\{ \pm2, \pm3\}$, $C=\{1\}$, gives the corresponding lumped Markov chain with the	
transition matrix 
\begin{equation}
	\tilde{P}_O =\frac{1}{4} \left(\begin{array}{ccc} 
		0 & 4 & 0 \\ 
		1 & 2 & 1 \\
		0 & 4 & 0
	\end{array}\right).
\end{equation}

The action of the evolution operator $U$ on the reduced space basis vectors is given by
	\begin{gather*}
		U(\ket{{A,B}}) = \ket{{B,A}}, \qquad	U(\ket{{B,A}}) = -\frac{1}{2} \ket{{A,B}} + \frac{1}{\sqrt{2}} \ket{{B,B}} + \frac{1}{2} \ket{{C,B}}, \\
		U(\ket{{B,B}}) = \frac{1}{\sqrt{2}}\ket{{A,B}} + \frac{1}{\sqrt{2}}\ket{{C,B}}, \;	U(\ket{{B,C}}) = \frac{1}{2} \ket{{A,B}} + \frac{1}{\sqrt{2}} \ket{{B,B}} - \frac{1}{2} \ket{{C,B}}, \;		U (\ket{{C,B}}) = \ket{{B,C}}. 
	\end{gather*}

The  Verblunsky coefficients and the corresponding CMV orthonormal basis of the reduced space for the quantized chain with the initial state $e_0 =	 \ket{\phi_{-1}} = \ket{A,B}$ are given by
\begin{equation*}
\alpha_0 = 0, \quad 	\alpha_1 = -\frac{1}{2}, \quad \alpha_2 = \frac{2}{3},
\quad 	\alpha_3 = \frac{1}{5},  \quad \alpha_4=1,
\end{equation*}
and
\begin{gather*}
	  e_1=\ket{B,A}, \quad
	 e_2=\frac{1}{\sqrt{3}}\left(\sqrt{2}\ket{B,B} + \ket{B,C}\right), \quad	
	 e_3=\frac{1}{\sqrt{15}}\left(\sqrt{2}\ket{B,B} - 2\ket{B,C} + 3 \ket{C,B}\right), \\  e_4=\frac{1}{\sqrt{5}}\left(-\ket{B,B} + \sqrt{2}\ket{B,C} + \sqrt{2} \ket{C,B}\right).
	\end{gather*}

\subsection{Icosahedron}
The vertex set of the icosahedron graph $I$, visualized on Figure~\ref{fig:icosahedron}, consists of ordered pairs of different elements of the set $\{1,2,3,4\}$. Two distinct vertices $(i,j)$ and $(k,l)$ are connected by an edge  if $i=k$ or $j=l$,  or $i,j,k,l$ are pairwise distinct and the sequence $(i,j,k,l)$ is an even permutation of $(1,2,3,4)$. 
\begin{figure}[h!]
	\begin{center}
		\includegraphics[width=6.5cm]{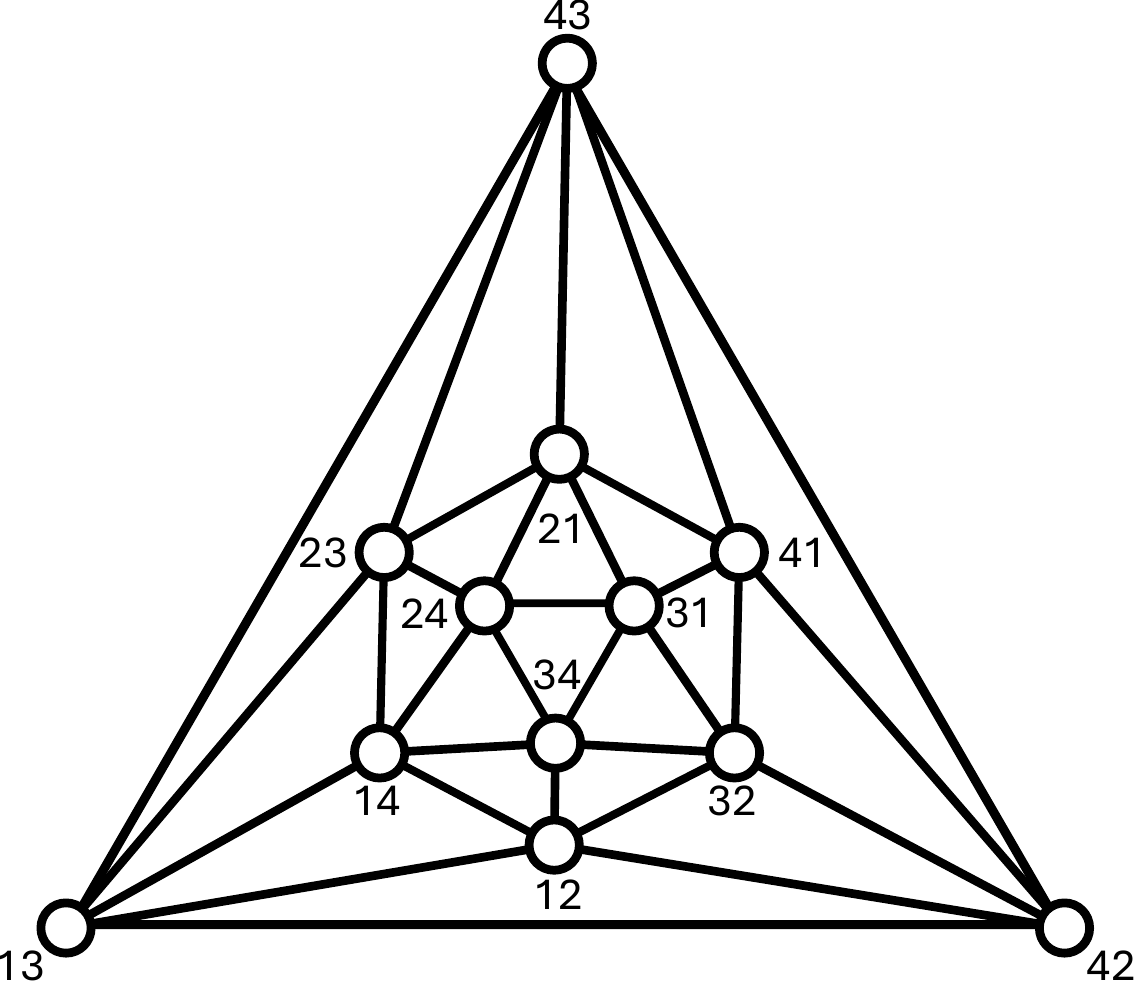}
	\end{center}
	\caption{The icosahedron graph}
	\label{fig:icosahedron}
\end{figure}	

Let us consider partition of the vertices into sets consisting of points with the same distance from the vertex $12$ 
\begin{gather*}
	A = \{ 12\}, \qquad 
	B = \{ 13, 14, 32, 34, 42 \}, \qquad
	C = \{ 31, 41, 23, 43, 24 \}, \qquad		 
	D = \{ 21 \}, 		 
\end{gather*}
then the transition matrix of the corresponding lumped Markov chain reads
\begin{equation}
	\tilde{P}_I =\frac{1}{5} \left(\begin{array}{cccc} 
		0 & 5 & 0 & 0 \\
		1 & 2 & 2 & 0 \\
		0 & 2 & 2 & 1 \\ 
		0 & 0 & 5 & 0  
	\end{array}\right).
\end{equation}
The action of the evolution operator $U$ in the reduced space is given by
\begin{gather*}
	U(\ket{{A,B}}) = \ket{{B,A}}, \qquad 
	U(\ket{{B,A}}) = -\frac{3}{5}\ket{{A,B}}+\frac{2\sqrt{2}}{5}\ket{{B,B}}+\frac{2\sqrt{2}}{5}\ket{{C,B}}, \\
	U(\ket{{B,B}}) = \frac{2\sqrt{2}}{5}\ket{{A,B}}-\frac{1}{5}\ket{{B,B}}+\frac{4}{5}\ket{{C,B}}, \qquad
	U(\ket{{B,C}}) =\frac{2\sqrt{2}}{5}\ket{{A,B}}+\frac{4}{5}\ket{{B,B}} - \frac{1}{5}\ket{{C,B}}, \\
	U(\ket{{C,C}}) = \frac{4}{5}\ket{{B,C}}-\frac{1}{5}\ket{{C,C}}+\frac{2\sqrt{2}}{5}\ket{{D,C}}, \qquad
	U(\ket{{C,B}}) = -\frac{1}{5}\ket{{B,C}}+\frac{4}{5}\ket{{C,C}}+\frac{2\sqrt{2}}{5}\ket{{D,C}}, \\
	U(\ket{{C,D}}) \frac{2\sqrt{2}}{5}\ket{{B,C}}+\frac{2\sqrt{2}}{5}\ket{{C,C}}-\frac{3}{5}\ket{{D,C}},
	\qquad
	U(\ket{{D,C}}) = \ket{{C,D}}.
\end{gather*}
The Verblunsky coefficients and the CMV basis for the initial state
$
	e_0 = \ket{\phi_{12}} = \ket{A,B}
$
read
\begin{equation*}
	\alpha_0 = 0, \quad \alpha_1 = -\frac{3}{5}, \quad \alpha_2 = \frac{1}{2}, \quad 
	\alpha_3 = -\frac{7}{15}, \quad \alpha_4 = \frac{8}{11}, \quad 	\alpha_5 = \frac{3}{19}, \quad \alpha_{6}=1,
\end{equation*}
\begin{gather*}
	e_1=\ket{B,A}, \quad
	e_2= \frac{1}{\sqrt{2}}\left(\ket{B,B} + \ket{B,C}\right),
	\quad
	e_3= \frac{1}{\sqrt{6}}\left(\ket{B,B} - \ket{B,C} + 2\ket{C,B}\right),
	\\
	e_4=\frac{\sqrt{33}}{11}\left(-\frac{\sqrt{2}}{3}\ket{B,B} + 
	\frac{\sqrt{2}}{3}\ket{B,C}+\frac{\sqrt{2}}{3}\ket{C,B} + 
	\sqrt{2}\ket{C,C}+  \ket{C,D}\right),
	\\
	e_5=\frac{\sqrt{209}}{209}\left(-\sqrt{2}\ket{B,B}+\sqrt{2}\ket{B,C}+
	\sqrt{2}\ket{C,B}+3\sqrt{2}\ket{C,C} -8\ket{C,D}+11\ket{D,C}\right),
	\\
	e_6=\frac{\sqrt{19}}{38}\left(\ket{B,B} - \ket{B,C} - \ket{C,B} - 3\ket{C,C}+4\sqrt{2}\ket{C,D}+ 4\sqrt{2}\ket{D,C}\right).
\end{gather*}

\subsection{Dodecahedron}
Vertices of the dodecahedron graph, visualized on Figure~\ref{fig:dodecahedron}, can be conveniently described as ordered pairs of different elements of the set $\{1,2,3,4,5\}$. Two vertices $(i,j)$ and $(k,l)$ are connected by an edge  if $i,j,k,l$ are pairwise distinct, and the sequence $(i,j,k,l,m)$ with $m$ being the unique complementary number is an even permutation of $(1,2,3,4,5)$. 
	\begin{figure}[h!]
		\begin{center}
			\includegraphics[width=5.5cm]{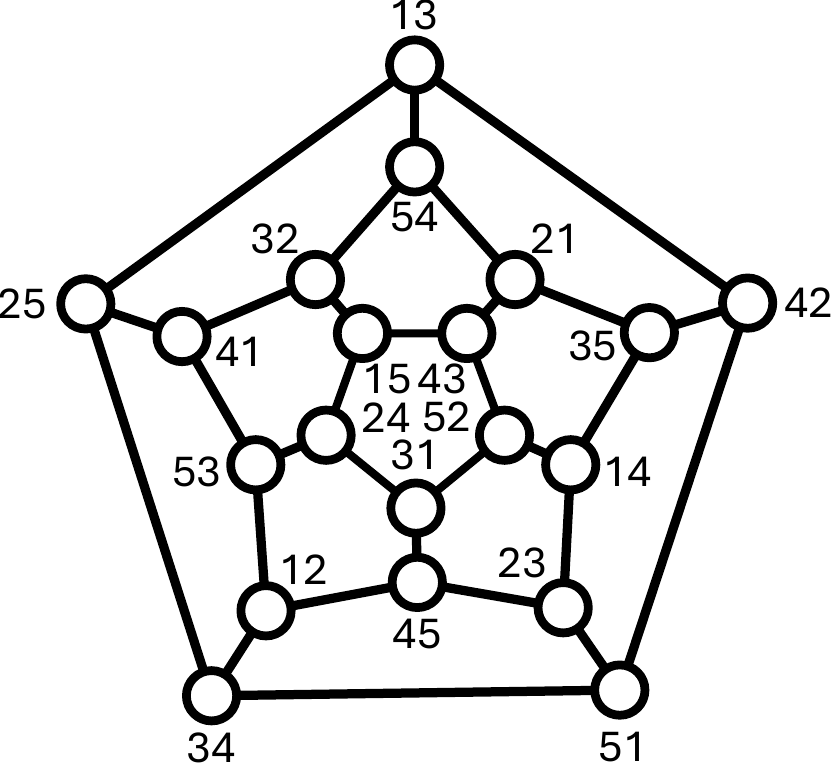}
		\end{center}
		\caption{The dodecahedron graph}
		\label{fig:dodecahedron}
	\end{figure}	

Consider partition of the vertices into sets consisting of points with the same distance from  $31$ 
\begin{gather*}
	A = \{ 31\}, \qquad 
	B = \{ 24, 52, 45 \}, \qquad
	C = \{ 53, 15, 43, 14, 12, 23 \}, \qquad
	D = \{ 41, 32, 21, 35, 34, 51 \}, \\
    E = \{ 25, 54, 42 \}, \qquad
	F = \{ 13 \}.
\end{gather*}
The transition matrix of the corresponding lumped Markov chain is given by
\begin{equation}
		\tilde{P}_D =\frac{1}{3} \left(\begin{array}{cccccc} 
			0 & 3 & 0 & 0 & 0 & 0 \\ 
			1 & 0 & 2 & 0 & 0 & 0 \\ 
			0 & 1 & 1 & 1 & 0 & 0 \\
			0 & 0 & 1 & 1 & 1 & 0 \\ 
			0 & 0 & 0 & 2 & 0 & 1 \\ 
			0 & 0 & 0 & 0 & 3 & 0 
		\end{array}\right).
	\end{equation}
The action of the evolution operator $U$ in the reduced space reads
	\begin{gather*}
		U(\ket{{A,B}}) = \ket{{B,A}}, \quad	U(\ket{{B,A}}) = -\frac{1}{3} \ket{{A,B}} + \frac{2\sqrt{2}}{3} \ket{{C,B}}, \quad	
		U(\ket{{B,C}}) = \frac{2\sqrt{2}}{3} \ket{{A,B}} + \frac{1}{3} \ket{{C,B}}, \\
		U(\ket{{C,C}}) = \frac{2}{3} \ket{{B,C}} - \frac{1}{3} \ket{{C,C}} + \frac{2}{3} \ket{{D,C}}, \qquad 
		U(\ket{{C,B}}) = -\frac{1}{3} \ket{{B,C}} + \frac{2}{3} \ket{{C,C}} + \frac{2}{3} \ket{{D,C}}, \\
		U(\ket{{C,D}})	= \frac{2}{3} \ket{{B,C}} + \frac{2}{3} \ket{{C,C}} - \frac{1}{3} \ket{{D,C}}, \qquad 
		U(\ket{{D,D}})	= \frac{2}{3} 	\ket{{C,D}} - \frac{1}{3} \ket{{D,D}} + \frac{2}{3} \ket{{E,D}}, \\
		U(\ket{{D,C}})	= -\frac{1}{3} \ket{{C,D}} + \frac{2}{3} \ket{{D,D}} + \frac{2}{3} \ket{{E,D}}, \qquad 
		U(\ket{{D,E}})	= \frac{2}{3}	\ket{{C,D}} + \frac{2}{3} \ket{{D,D}} - \frac{1}{3} \ket{{E,D}}, \\
		U(\ket{{E,D}})	= \frac{1}{3} \ket{{D,E}} + \frac{2\sqrt{2}}{3} \ket{{F,E}}, \quad  U( \ket{{E,F}})	= \frac{2\sqrt{2}}{3} \ket{{D,E}} - \frac{1}{3} \ket{{F,E}}, \quad U(\ket{{F,E}})	= \ket{{E,F}}.
	\end{gather*}

The Verblunsky coefficients for the initial state $e_0 = \ket{\phi_{31}} = \ket{A,B}$ read
\begin{gather*}
\alpha_0 = 0, \quad		\alpha_1 = -\frac{1}{3}, \quad \alpha_2 = 0, \quad 
\alpha_3 = -\frac{1}{3}, \quad \alpha_4 = \frac{1}{2}, \quad \alpha_5 = -\frac{5}{9}, \\
\alpha_6 = \frac{4}{7}, \quad \alpha_7 = -\frac{5}{33}, \quad \alpha_8 = \frac{8}{19}, \quad
	\alpha_9 = \frac{11}{27}, \quad \alpha_{10}=1,	
\end{gather*}	
and the vectors of the corresponding CMV basis in the reduced space are given by
		\begin{gather*}	
	e_0  = \ket{A,B} 	, \quad  e_1=\ket{B,A},
	\quad 	 e_2= \ket{B,C}, \quad  e_3= \ket{C,B}, \quad
	e_4=\frac{1}{\sqrt{2}}\left(\ket{C,C} + \ket{C,D}\right), \\
	e_5=\frac{1}{\sqrt{6}}\left(\ket{C,C} - \ket{C,D} +2\ket{D,C}\right), \quad
	e_6=\frac{1}{\sqrt{21}}\left(-\ket{C,C} + \ket{C,D} + \ket{D,C} + 3\ket{D,D} + 3 \ket{D,E}\right), \\
	e_7=\frac{1}{\sqrt{77}}\left(-\ket{C,C} + \ket{C,D} + \ket{D,C}+3\ket{D,D} -4\ket{D,E}+  7\ket{E,D}\right), \\
	e_8=\frac{1}{\sqrt{209}} \left( \sqrt{2} \left( \ket{C,C} - \ket{C,D} - \ket{D,C} -3\ket{D,D} +   4\ket{D,E}+4\ket{E,D}\right) +11\ket{E,F} \right), \\
	e_9=\frac{\sqrt{57}}{171}\left(\sqrt{2} \left( \ket{C,C} - \ket{C,D} - \ket{D,C} -3\ket{D,D}+4\ket{D,E} +   4\ket{E,D} \right) -8\ket{E,F}+19\ket{F,E} \right),\\
	e_{10}=\frac{\sqrt{3}}{18} \left( -\ket{C,C} + \ket{C,D} +\ket{D,C}+3\ket{D,D}-4\ket{D,E}-   4\ket{E,D}+ 4\sqrt{2}\ket{E,F}+4\sqrt{2}\ket{F,E} \right) .
\end{gather*}

\begin{Rem} In this Section we studied quantum walks corresponding to lumped random walks on Platonic graphs with respect to certain well defined partitions. The coarse-grained graphs turned out to be paths, and the same reduced quantum walk was obtained by application of CMV uniformization algorithm for appropriate initial vector. Indeed, using equations \eqref{eq:rk}-\eqref{eq:pk} one can check that the Verblunsky coefficients, calculated in the reduction process, give the lumped transition matrices. 
	\end{Rem}

\section{Quantum walk on the hypercube graph, and the quantum Ehrenfests model} \label{sec:q-QN-E}
In this Section we present a generalization of our results on quantum walk on hexahedron and its reduction to arbitrary dimension. 
Consider the $N$-dimensional hypercube graph $Q_N$ with the vertex set $V(Q_N)=\ZZ_2^N$ consisting of binary sequences of length $N$, the edge set $E(Q_N)$ consists of (ordered) pairs of vertices/sequences which differ in one element only. 
\begin{Rem}
	It may be useful to consider the hypercube $Q_N$ as the Cayley graph of the Abelian group $(\ZZ_2^N, \oplus)$ with the set of generators $\{ e_j \}_{j=1}^N$ being binary sequences with one at $j$-th place and zeros at all other places.
\end{Rem}

Denote by $A_{k}$, $k=0,1,\dots ,N$, the set of binary sequences (of length $N$) with exactly $k$ ones. The homogeneous random walk on the hypercube is strongly lumpable with respect to such partition. 
 The corresponding lumped chain has the structure of the random walk on the path $\{ 0, 1, \dots , N\}$ with the following transition probabilities 
\begin{equation} \label{eq:transitions-E}
	q_k = \frac{k}{N}, \qquad r_k = 0, \qquad p_k = 1 - \frac{k}{N}, \qquad k=0,1, \dots N.
\end{equation}
The above aggregation result is well known~\cite{Diaconis}, and the lumped chain has an interpretation as the Ehrenfests~\cite{Ehrenfest} urn model. Recall that in the Ehrenfests model one considers two urns and $N$ balls distributed in the urns. The system is said to be in state $k$ if there are $k$ balls in the first urn, $N-k$ balls in the second urn. At a discrete-time step a ball is chosen uniformly at random and moved to the other urn. 

On the quantum level, the vectors which span the reduced space within the Hilbert space of the quantum hypercube walk, are given by
\begin{gather*}
	\ket{k,k+1} = \sqrt{\frac{k! (N-k-1)!}{N!}} \sum_{\substack{u\in A_k, v\in A_{k+1}\\ (u,v)\in E(Q_N)}} \ket{u}\otimes \ket{v} , \qquad k = 0,1, \dots ,N-1, \\
		\ket{k,k-1} = \sqrt{\frac{(k-1)! (N-k)!}{N!}} \sum_{\substack{u\in A_k, v\in A_{k-1} \\ (u,v)\in E(Q_N)}} \ket{u}\otimes \ket{v} , \qquad k = 1, 2,\dots , N.
\end{gather*}
Under action of the quantum evolution operator on the hypercube they change according to
\begin{gather*}
U(\ket{k,k+1}) = \frac{2\sqrt{k(N-k)}}{N} \ket{k-1,k} + \frac{N-2k}{N} \ket{k+1,k}, \\
U(\ket{k,k-1}) = \frac{2k-N}{N} \ket{k-1,k} +  \frac{2\sqrt{k(N-k)}}{N} \ket{k+1, k} .
\end{gather*}
This is the same behavior like in the Szegedy quantization of the Ehrenfests urn model with the natural basis 
\begin{equation*}
	\ket{0}\otimes \ket{1}, \qquad \ket{k}\otimes\ket{k\pm1},  \qquad \ket{N}\otimes \ket{N-1}, \qquad k=1,\dots ,N-1,
\end{equation*}
and the reflection in the hyperspace spanned by the vectors
\begin{equation*}
	\ket{\phi_k} = \frac{1}{\sqrt{N}}\left( \sqrt{k} \ket{k}\otimes\ket{k-1} + \sqrt{N-k} \ket{k} \otimes \ket{k+1}\right), \qquad k=0,1,\dots ,N.
\end{equation*}
\begin{Rem}
The spectral properties of the Ehrenfests model are well known, and they follow from relation of the model~\cite{KarlinMcGregor} to the Kravtchouk polynomials~\cite{Nikiforov-Suslov-Uvarov}. In particular, the stationary state of the model is provided by the binomial distribution, and the eigenvalues of the transition matrix are uniformly distributed on the segment $[-1,1]$. Discrete-time quantum walk on a hypercube in connection to the Ehrenfests model and Kravtchouk polynomials is discussed in \cite{HoKonno}, and its continuous time version is studied in \cite{Ekert,Ekert2} in relation to perfect state transfer. 
\end{Rem}
Notice that, unlikely the natural basis of states in the quantization of the Ehrenfests model, the corresponding states of the quantum walk on the hypercube $Q_N$ are not (except of $\ket{0,1}$ and $\ket{N,N-1}$) product states. This can be checked by calculating the von Neumann entropy of the reduced states. For example, direct calculation, by tracing  with respect of the second factor of the state $|k,k+1\rangle \bra{k,k+1}$,  gives the matrix of order $\binom{N}{k}$ 
	\begin{equation*}
\text{tr}_C \rho_{|k,k+1\rangle}  = \!\!\!
\sum_{w\in A_{k+1}} \! \! \! \mathrm{Id}_1 \otimes \braket{w|k,k+1} \braket{k,k+1|\mathrm{Id}_1 \otimes|w}
= \frac{k! (N-k-1)!}{N!}\begin{pmatrix}
	N-k & 1 & \cdots & 1 \\
	1 & N-k & \ddots & \vdots \\
	\vdots & \ddots & \ddots &1\\
	1 & \cdots & 1 & N-k
\end{pmatrix}.
\end{equation*}	
Its eigenvalues are
\begin{align*}
	\frac{N-k-1}{\binom{N}{k} (N-k)} & \qquad \text{with multiplicity} \quad \binom{N}{k} -1 ,\\
	\frac{\binom{N}{k} + N-k-1}{\binom{N}{k} (N-k)} & \qquad \text{with multiplicity} \quad 1 ,
\end{align*}
what implies that the corresponding von Neumann entropy does not vanish.

Finally, let us present (the proof is by induction, but in fact it is a special case of the general theory given in~\cite{Doliwa-Siemaszko-QW}) the Verblunsky coefficients and the corresponding CMV basis of the quantum Ehrenfests model, for the initial vector $e_0 = \ket{0} \otimes \ket{1}$ 
\begin{align*}
	\alpha_{2k} = & 0,  & e_{2k} & = \ket{k}\otimes \ket{k+1}, \qquad k=0,1,\dots , N-1,\\
	\alpha_{2k-1} = & \frac{2k-N}{N} ,  & e_{2k-1} & = \ket{k}\otimes \ket{k-1}, \qquad k=1,2, \dots N.
\end{align*}
By equations \eqref{eq:rk}-\eqref{eq:pk} one recovers the transition probabilities \eqref{eq:transitions-E}. The same result, modulo the correspondence $\ket{k}\otimes \ket{k\pm 1} \leftrightarrow \ket{k,k\pm 1}$, is obtained by the CMV algorithm for the unitary evolution operator on the hypercube and the initial vector $e_0 = \ket {\phi_{00\dots 0}}$.

\section{Aggregated quantum walks on a free group}
In this final Section we extend our approach to graphs with infinite set of vertices, but with very regular structure. On example of the Cayley graph of the free group with two generators we describe strong lumping which leads to quantum walks on half line with constant transition probabilities. As any group can be obtained by imposing certain relations on generators of a free group then such graphs are potentially important in studying quantum walks on graphs of other groups.

Consider free group over two generators $S = \{ a,b \}$ and its Cayley graph $CF_2$. Recall that vertices $V(CF_2)$ of the graph are (reduced) words $w$ over alphabet  with letters $a^{\pm1}$, $b^{\pm 1}$ of $S\cup S^{-1}$ including the empty word $e$ being the identity element. Two vertices belong to the edge set $(u,v)\in E(CF_2)$ when $u^{-1}v \in S\cup S^{-1}$. The transition matrix of the corresponding random walk is proportional to the adjacency matrix $A(CF_2)$ with homogeneous probabilities equal to $1/4$. We consider partition of $V(CF_2)$ into subsets $A_k$ containing vertices of the same distance $k$ from the identity $e$, i.e. $A_0 = \{ e\}$, $A_1 = \{ a, a^{-1}, b, b^{-1}\}$, and
$
	A_2 = \{ aa, ab, ab^{-1}, a^{-1} a^{-1}, a^{-1} b, a^{-1} b^{-1}, bb, ba, ba^{-1}, b^{-1}b^{-1}, b^{-1} a, b^{-1} a^{-1} \}$;
we have $|A_k| = 4 \cdot 3^{k-1}$ words of length $k\geq 1$. 
 
The corresponding reduced walk on $\NN_0$ has transition probabilities 
\begin{equation}
	p_0 = 1, \qquad p_k = \frac{3}{4}, \qquad q_k = \frac{1}{4}, \qquad k \geq 1.
\end{equation}
Szegedy's quantization of such random walk is a special case of quantum walks with two-periodic Verblunsky coefficients studied in~\cite{Doliwa-Siemaszko-QW}. In particular, its Verblunsky coefficients and the CMV basis read
\begin{equation*}
	\alpha_{2k} = 0, \qquad \alpha_{2k+1} = -\frac{1}{2}, \qquad 
	e_{2k} = \ket{k}\otimes \ket{k+1}, \qquad e_{2k+1} = \ket{k+1}\otimes \ket{k}, \qquad
	k = 0,1,2, \dots \; .
\end{equation*}
\begin{Rem}
The above calculations easily generalize to an arbitrary free group with finite number of generators. \end{Rem}

One can study other groups by imposing certain relations on the generators. For example, the hexahedron graph is a Cayley graph of the group obtained by imposing on the set of three generators $\{a,b,c\}$ the involutivity $a^2 = b^2 = c^2= e$ and commutativity $ab = ba$, $ac = ca$, $bc = cb$ relations. Relaxing the involutivity condition we obtain the Cayley graph having the structure of three dimensional integer lattice $\ZZ^3$.

Consider the second case, where we assume involutivity of the generators only. Then the vertices of the corresponding Cayley graph can be identified with the words over the alphabet $\{a,b,c\}$ with no occurrence of the same two consecutive letters. Splitting again the vertices into subsets $A_k$, $k=0,1,2,\dots $, having the same distance from the identity element $e$, i.e., $A_0=\{ e\}$, $A_1=\{ a,b,c\}$, $A_2 = \{ ab, ac, ba, bc, ca, cb \}$, we obtain $|A_k| = 3\cdot 2^{k-1}$, $k\geq 1$. 
The corresponding reduced walk on $\NN_0$ has transition probabilities 
\begin{equation}
	p_0 = 1, \qquad p_k = \frac{2}{3}, \qquad q_k = \frac{1}{3}, \qquad k \geq 1.
\end{equation}
Its Verblunsky coefficients and the CMV basis read
\begin{equation*}
	\alpha_{2k} = 0, \qquad \alpha_{2k+1} = -\frac{1}{3}, \qquad 
	e_{2k} = \ket{k}\otimes \ket{k+1}, \qquad e_{2k+1} = \ket{k+1}\otimes \ket{k}, \qquad
	k = 0,1,2, \dots \; .
\end{equation*}
\begin{Rem}
	As it was shown in \cite{Doliwa-Siemaszko-QW}, random walks on $\NN_0$ with two-periodic Verblunsky coefficients are described by Chebyshev orthogonal polynomials. Our examples of this Section provide special cases of such a relation. 
\end{Rem}

\section{Conclusion and open problems}
We presented aggregation procedure of Szegedy's quantum walk on a graph allowing for a coarse-graining reduction.  Given the possibility of two operations on a Markov chain -- aggregation and quantization -- a natural question arises regarding their permutability. The possibility of such an exchange can be of great practical significance in developing new algorithms or investigating the properties of graphs. In order for both operations to commute, the initial Markov chain should satisfy certain conditions in relation to the underlying partition of states, which includes the strong lumpability property.  This condition must be supplemented by nonlinear relations, generalizing Kolmogorov's cycle condition, that involve the coefficients of both the original and the aggregated transition matrices. 

 Such conditions are fulfilled, for example, in the case of homogeneous random walks on distance-regular graphs, which  form a class of regular graphs with strong combinatorial symmetry, and their partition into spherical sets. We gave several examples including aggregation of quantum walks on Platonic solids, hypercube graphs or Cayley graphs of free groups. We have only just touched upon the quantization of walks on graphs equipped with equitable partitions. Due to the importance of this concept in algebraic graph theory and its relation with other branches of mathematics and computer science this connection certainly deserves deeper studies.

The CMV uniformization of unitary operators naturally leads to operators which can be interpreted as quantum walks on path graphs and the corresponding decomposition of the Hilbert space. Because the coarse-grained graph may not have the path topology this leads to interesting observations regarding the "straightening" procedure of Markov chains. Preliminary calculation shows relation of the above to old questions \cite{Dirac,Bartlett,Feynman} concerning appearance of "negative probability" in quantum mechanics. In connection to recent application of the quasi-probability techniques in machine learning, mathematical finance or quantum information theory~\cite{Blass-Gurevich,TemmeBravyiGambetta,PiveteauSutterWoerner} this issue deserves deeper studies and will be reported in~\cite{DoliwaSiemaszkoZalewski-quasi}.

 Another aspects worth of study, mentioned in the body of the paper, are reductions of walks on Cayley graphs implied by relation on the generators of free groups, information-theoretic aspect of aggregation of quantum walks or the connection to the theory of integrable systems. 
 	Because many computational problems can be expressed in terms of graph exploration, quantum walks continue to be one of the most active areas of quantum algorithm research. A graph with a large automorphism group often allows the Hilbert space to be decomposed into invariant subspaces. These questions lie at the intersection of quantum algorithms, algebraic graph theory, spectral graph theory, and representation theory. They are mathematically rich, have direct connections to open problems such as graph isomorphism, and are well suited to both theoretical analysis and numerical experimentation.
 
  One can view quantum walks as the single-particle sector of a quantum many-body system, with spin chains providing one of the simplest realizations. The challenge is to understand how the elegant graph-theoretic picture changes when moving from one excitation to many interacting excitations. The present research can be generalized to compare interacting and non-interacting walks on graph families equipped with various symmetry structures. It is anticipated that results from the theory of integrable quantum systems~\cite{QISM} may prove useful in these studies.
 
 Finally we remark that in our article we have restricted attention to aggregation of quantum walks obtained by Szegedy quantization of discrete-time Markov chains. Other quantization schemes can be investigated from the coarse-graining point of view as well. Moreover, aggregation of quantum walks which do not have direct classical counterparts, and which can be used to solve various computational problems is of considerable interest.

 \section*{Acknowledgements}
 We would like to thank the anonymous Reviewers, whose comments and suggestions helped to improve the presentation of the results and correct an error in the original formulation of Proposition~\ref{prop:P}.

\end{document}